\documentclass[11pt]{article}
\usepackage{amsfonts}
\usepackage{amssymb}
\usepackage{amsmath}
\usepackage{graphicx}
\usepackage{geometry}

\newtheorem{theorem}{Theorem}

\newtheorem{corollary}{Corollary}

\newtheorem{proposition}{Proposition}

\newenvironment{proof}[1][Proof]{\textbf{#1.} }{\ \rule{0.5em}{0.5em}}
\makeatletter
\def\@removefromreset#1#2{\let\@tempb\@elt
     \def\@tempa#1{@&#1}\expandafter\let\csname @*#1*\endcsname\@tempa
     \def\@elt##1{\expandafter\ifx\csname @*##1*\endcsname\@tempa\else
    \noexpand\@elt{##1}\fi}     \expandafter\edef\csname cl@#2\endcsname{\csname cl@#2\endcsname}     \let\@elt\@tempb
     \expandafter\let\csname @*#1*\endcsname\@undefined}

\@removefromreset{equation}{section}

\@removefromreset{theorem}{section}
\makeatother
\input{tcilatex}
\begin{document}

\title{Local hidden variable modelling, classicality, quantum separability,
and the original Bell inequality}
\author{Elena R. Loubenets\bigskip \\
%EndAName
Applied Mathematics Department, Moscow State Institute \\
of Electronics and Mathematics, Moscow 109028, Russia}
\date{}
\maketitle

\begin{abstract}
We introduce a general condition sufficient\emph{\ }for the validity of the
original Bell inequality (1964) in a local hidden variable (LHV) frame. This
condition can be checked experimentally and incorporates only as a
particular case the assumption on perfect correlations or anticorrelations
usually argued for this inequality in the literature. Specifying this
general condition for a quantum bipartite case, we introduce the whole class
of bipartite quantum states, separable and nonseparable, that (i) admit an
LHV description under any bipartite measurements with two settings per site;
(ii) do not necessarily exhibit perfect correlations and may even have a
negative correlation function if the same quantum observable is measured at
both sites, but (iii) satisfy the "perfect correlation" version of the
original Bell inequality for any three bounded quantum observables $A_{1},$ $%
A_{2}=B_{1},$ $B_{2}$ at sites "A" and "B", respectively. Analysing the
validity of this general LHV condition under classical and quantum
correlation scenarios with the same physical context, we stress that, unlike
the Clauser-Horne-Shimony-Holt (CHSH) inequality, the original Bell
inequality distinguishes between classicality and quantum separability.
\end{abstract}

\tableofcontents

\section{Introduction}

Analysing in 1964 a possibility of a local hidden variable (LHV) description
of bipartite\footnote{%
In quantum information, two parties (observers) are usually named as Alice
and Bob.} quantum measurements on two-qubits, J. Bell introduced \cite{1}
the LHV\ constraint on correlations, usually now referred to as \emph{the
original Bell inequality.} Both Bell's proofs \cite{1, 2} of this LHV
inequality are essentially built up on two additional assumptions - a
dichotomic character of Alice's and Bob's measurements plus the perfect
correlation or anticorrelation of Alice's and Bob's outcomes for a definite
pair of their local settings. Specifically, the latter assumption is usually
abbreviated as the assumption on perfect correlations or anticorrelations.

Bell's proofs are now reproduced in any textbook on quantum information and
there exists the widespread misconception\footnote{%
For this misleading opinion, see, for example, the corresponding paragraph
of the Wikipedia article on Bell's theorem.} that, in any LHV case, the
original Bell inequality cannot hold without the additional assumptions used
by Bell \cite{1, 2}, specifically, without the assumption on perfect
correlations or anticorrelations - as it has been, for example, claimed by
Simon \cite{3} and Zukowski \cite{4}.

However, Bell's additional assumptions are \emph{only} \emph{sufficient }but
not necessary for the validity of the original Bell inequality in an LHV
frame. Based on operator methods, we, for example, proved in [5 - 8] that,
in either of bipartite cases, classical\footnote{%
Throughout the present paper, the term "classical" is meant in its physical
sense.} or quantum, the original Bell inequality holds for Alice's and Bob's
real-valued outcomes in $[-1,1]$ of any spectral type, continuous or
discrete, not necessarily dichotomic, and that there exist bipartite quantum
states $\rho $ on a Hilbert space $\mathcal{H}\otimes \mathcal{H}$ that do
not necessarily exhibit perfect correlations and may even have [8] a
negative correlation if the same quantum observable is measured at both
sites but satisfy the "perfect correlation" (minus sign) version of the
original Bell inequality\footnote{%
For correlations of an arbitrary type, not necessarily quantum, the original
Bell inequality, in its "perfect correlation/anticorrelation" version has
the form (\ref{9}).}:%
\begin{equation}
\left\vert \text{ }\mathrm{tr}[\rho \{A_{1}\otimes B_{1}\}]-\mathrm{tr}[\rho
\{A_{1}\otimes B_{2}\}]\right\vert \leq 1-\mathrm{tr}[\rho \{B_{1}\otimes
B_{2}\}],  \label{xx}
\end{equation}%
for \emph{any} three bounded quantum observables $A_{1},$ $A_{2}=B_{1},$ $%
B_{2}$ on $\mathcal{H}$, measured at sites of Alice and Bob, respectively,
and with eigenvalues in $[-1,1]$.

Note that if bipartite measurements, with two settings per site, are
performed on any of bipartite quantum states considered in [5 - 8], then
these bipartite quantum measurements \emph{admit} \emph{an LHV description}.
For a nonseparable quantum state, this fact follows from theorem 1 in \cite%
{9} and the existence (see Eq. (A5) in \cite{8}) for these bipartite
measurements of a joint probability measure returning all observed joint
distributions as marginals.

Our results in [5 - 8] clearly indicate that, in an LHV frame, the original
Bell inequality holds for outcomes of any spectral type, not necessarily
dichotomic, and under an additional assumption which is more general than
the assumption on perfect correlations or anticorrelations and ensures the
existence, in a quantum LHV case, of bipartite quantum states never
violating the original Bell inequality (\ref{xx}).

The aim of the present paper is to find such a general LHV condition and to
specify the validity of this general condition under bipartite correlation
scenarios on classical and quantum particles. The paper is organized as
follows.

In section 2, for an arbitrary bipartite correlation scenario, with two
settings per site and real-valued outcomes of any spectral type, discrete or
continuous, we introduce a new condition sufficient for the validity of the
original Bell inequality in an LHV frame. We prove that, in an LHV\ model of
any type, this new LHV condition is more general than the assumption on
perfect correlations or anticorrelations and incorporates the latter only as
a particular case. We show that, in a dichotomic case, the new general LHV\
condition can be tested experimentally.

In section 3,\ we stress that, under a bipartite correlation experiment,
performed on classical particles and described by bounded classical
observables $A_{1},$ $A_{2}=B_{1},$ $B_{2}$ at sites of Alice and Bob,
respectively, the minus sign version of the new LHV condition and,
therefore, the "perfect correlation" version of the original Bell inequality
are fulfilled for any initial state of classical particles and any type of
Alice's and Bob's classical measurements - ideal (necessarily exhibiting
perfect correlations if the same classical observable is measured at both
sites) or non-ideal (not necessarily exhibiting perfect correlations).

In section 4, for a quantum correlation scenario, we specify the new general
LHV condition in quantum terms. This allows us to introduce a new class of
bipartite quantum states that admit an LHV description under any bipartite
quantum measurements with two settings per site and satisfy the "perfect
correlation" version of the original Bell inequality for any bounded quantum
observables $A_{1},$ $A_{2}=B_{1},$ $B_{2}$ at sites of Alice and Bob,
respectively. These quantum states do not necessarily exhibit perfect
correlations and may even have a negative correlation function whenever the
same quantum observable is measured at both sites. We stress that an
arbitrary separable quantum state does not need to belong to this new class
of bipartite quantum states and that all bipartite quantum states specified
by us earlier in [5 - 8] are included into this new state class only as a
particular subclass.

In section 5, we summarize the main results of the present paper and discuss
their significance for the statistical analysis of correlation experiments
on classical and quantum particles. We stress that these results rigorously
disprove the widespread misconceptions existing in the literature ever since
1964, in particular, the recent claims of Simon \cite{3} and Zukowski \cite%
{4} on the relation between the original Bell inequality, perfect
correlations, classicality and quantum separability.

\section{General bipartite case}

Consider the probabilistic description\footnote{%
For the general framework on the probabilistic description of a multipartite
correlation scenario, see \cite{9}.} of a bipartite correlation scenario
("gedanken" experiment), specified at Alice's and Bob's sites by measurement
settings $a_{i},$ $b_{k},$ $i,k=1,2,$ and real-valued outcomes $\lambda
_{1},\lambda _{2}\in \lbrack -1,1]$, respectively. This correlation scenario
is described by four bipartite joint measurements $(a_{i},b_{k}),$ $i,k=1,2,$
with joint probability distributions $P^{(a_{i},b_{k})}$.

For a joint measurement $(a_{i},b_{k}),$ denote by 
\begin{equation}
\langle \lambda _{n}\rangle ^{(a_{i},b_{k})}:=\dint \lambda _{n}\text{ }%
P^{(a_{i},b_{k})}(\mathrm{d}\lambda _{1}\times \mathrm{d}\lambda _{2}),\text{
\ \ }n=1,2,  \label{3}
\end{equation}%
the mean values of outcomes at Alice's ($n=1$) and Bob's ($n=2$) sites and
by 
\begin{equation}
\langle \lambda _{1}\lambda _{2}\rangle ^{(a_{i},b_{k})}:=\dint \lambda
_{1}\lambda _{2}P^{(a_{i},b_{k})}(\mathrm{d}\lambda _{1}\times \mathrm{d}%
\lambda _{2})  \label{4}
\end{equation}%
the expected value of the product of their outcomes. In quantum information,
this product expectation is usually referred to as a \emph{correlation
function }or \emph{correlation}, for short.

If, under a joint measurement $(a_{i},b_{k}),$ Alice's and Bob's outcomes
are perfectly correlated or anticorrelated in the sense that either event $%
\{\lambda _{1}=\lambda _{2}\}$ or event $\{\lambda _{1}=-\lambda _{2}\}$ are
observed with certainty, then%
\begin{eqnarray}
P^{(a_{i},b_{k})}(\{\lambda _{1} &=&\lambda _{2}\})=\dint\limits_{\lambda
_{1}=\lambda _{2}}P^{(a_{i},b_{k})}(\mathrm{d}\lambda _{1}\times \mathrm{d}%
\lambda _{2})=1  \label{5} \\
&&\text{or}  \notag \\
P^{(a_{i},b_{k})}(\{\lambda _{1} &=&-\lambda _{2}\})=\dint\limits_{\lambda
_{1}=-\lambda _{2}}P^{(a_{i},b_{k})}(\mathrm{d}\lambda _{1}\times \mathrm{d}%
\lambda _{2})=1,  \notag
\end{eqnarray}%
respectively. For two-valued outcomes $\lambda _{1},\lambda _{2}=\pm 1,$
this perfect correlation/anticorrelation condition is equivalently
represented by the restriction on the correlation function: $\langle \lambda
_{1}\lambda _{2}\rangle ^{(a_{i},b_{k})}$ $=\pm 1,$ introduced originally by
Bell \cite{1, 2}.

The following theorem introduces a new condition sufficient for the validity
of the original Bell inequality in an arbitrary LHV model. As it is further
shown below, this new LHV\ condition is more general than the assumption on
perfect correlations or anticorrelations and incorporates the latter only as
a particular case.

\begin{theorem}
Let a $2\times 2$ - setting bipartite correlation scenario, with outcomes $%
\lambda _{1},\lambda _{2}\in \lbrack -1,1]$ of an arbitrary spectral type,
discrete or continuous, admit an LHV model for correlation functions, that
is, each $\langle \lambda _{1}\lambda _{2}\rangle ^{(a_{i},b_{k})},$ $%
i,k=1,2,$ admits representation%
\begin{equation}
\langle \lambda _{1}\lambda _{2}\rangle
^{(a_{i},b_{k})}=\dint\limits_{\Omega }f_{1}^{(a_{i})}(\omega
)f_{2}^{(b_{k})}(\omega )\text{ }\nu (\mathrm{d}\omega )  \label{7}
\end{equation}%
in terms of some variables $\omega \in \Omega ,$ a probability distribution $%
\nu $ of these variables and real-valued functions $f_{1}^{(a_{i})}(\omega
), $ $f_{2}^{(b_{k})}(\omega )\in $ $[-1,1].$ If, in this LHV model,
condition%
\begin{equation}
\dint\limits_{\Omega }f_{2}^{(b_{2})}(\omega )\text{ }\left(
f_{2}^{(b_{1})}(\omega )\mp f_{1}^{(a_{2})}(\omega )\right) \nu (\mathrm{d}%
\omega )\geq 0  \label{8}
\end{equation}%
is fulfilled in its minus sign or plus sign form, then the original Bell
inequality%
\begin{equation}
\left\vert \langle \lambda _{1}\lambda _{2}\rangle ^{(a_{1},b_{1})}-\langle
\lambda _{1}\lambda _{2}\rangle ^{(a_{1},b_{2})}\right\vert \leq 1\mp
\langle \lambda _{1}\lambda _{2}\rangle ^{(a_{2},b_{2})}  \label{9}
\end{equation}%
holds in its minus sign ("perfect correlation") or plus sign ("perfect
anticorrelation") version, respectively.
\end{theorem}

\begin{proof}
In view of representation (\ref{7}),%
\begin{eqnarray}
&&\langle \lambda _{1}\lambda _{2}\rangle ^{(a_{1},b_{1})}-\langle \lambda
_{1}\lambda _{2}\rangle ^{(a_{1},b_{2})}  \label{10} \\
&=&\dint\limits_{\Omega }f_{1}^{(a_{1})}(\omega )\left\{
f_{2}^{(b_{1})}(\omega )-f_{2}^{(b_{2})}(\omega )\right\} \nu (\mathrm{d}%
\omega ).  \notag
\end{eqnarray}%
From relation (\ref{10}), the number inequality $\left\vert x-y\right\vert
\leq 1-xy,$ $\forall x,y\in \lbrack -1,1],$ and condition (\ref{8}) it
follows:%
\begin{eqnarray}
&&\left\vert \langle \lambda _{1}\lambda _{2}\rangle
^{(a_{1},b_{1})}-\langle \lambda _{1}\lambda _{2}\rangle
^{(a_{1},b_{2})}\right\vert  \label{11} \\
&\leq &\dint\limits_{\Omega }\left\vert f_{2}^{(b_{1})}(\omega
)-f_{2}^{(b_{2})}(\omega )\right\vert \text{ }\nu (\mathrm{d}\omega )  \notag
\\
&\leq &\dint\limits_{\Omega }\left( 1-f_{2}^{(b_{1})}(\omega
)f_{2}^{(b_{2})}(\omega )\right) \text{ }\nu (\mathrm{d}\omega )  \notag \\
&\leq &1\mp \langle \lambda _{1}\lambda _{2}\rangle ^{(a_{2},b_{2})},  \notag
\end{eqnarray}%
where the minus (or plus) sign form of condition (\ref{8}) corresponds to
the minus (or plus) sign version of relation (\ref{11}). This proves the
statement.
\end{proof}

\bigskip

At the end of this section, we show that, in a dichotomic case, the LHV
condition (\ref{8}) can be tested experimentally. Note that Bell's
assumption $\langle \lambda _{1}\lambda _{2}\rangle ^{(a_{2},b_{1})}=\pm 1$
also refers only to a dichotomic case.

According to the terminology introduced by Fine \cite{10}, a correlation LHV
model is referred to as \emph{deterministic} if the values of functions $%
f_{1}^{(a_{i})},$ $f_{2}^{(b_{k})}$, $i,k=1,2,$ coincide with Alice's and
Bob's outcomes under their corresponding measurements and \emph{stochastic}
- otherwise. If, in addition, functions $f_{1}^{(a_{i})},$ $f_{2}^{(b_{k})},$
$i,k=1,2,$ are conditioned by any extra relation, then we refer \cite{9} to
such a correlation LHV model as \emph{conditional.}

Therefore, otherwise expressed, theorem 1 reads - \emph{if a }$2\times 2$%
\emph{-setting correlation scenario, with outcomes }$\lambda _{1},\lambda
_{2}\in \lbrack -1,1]$\emph{\ of any spectral type, admits a conditional
correlation LHV model }(\ref{7}), (\ref{8})\emph{, then the original Bell
inequality }(\ref{9})\emph{\ holds.}

The original Bell inequality (\ref{9}) represents an example of a
conditional Bell-type inequality\footnote{%
For the definition of a Bell-type inequality, conditional or unconditional,
see \cite{11}.}.

We stress that the LHV condition (\ref{8}) \emph{does not, in general,} 
\emph{imply} \emph{any restriction} on a value of correlation $\langle
\lambda _{1}\lambda _{2}\rangle ^{(a_{2},b_{1})}$. This, in particular,
means that, in contrast to the claims of Simon and Zukowski in \cite{3, 4},
for the "perfect correlation" (minus sign) form of inequality (\ref{9}) to
hold in an LHV\ frame, expectation $\langle \lambda _{1}\lambda _{2}\rangle
^{(a_{2},b_{1})}$ does not need to be even positive (see example 3 in
section 4.1).

Condition (\ref{8}) is, in particular, fulfilled if 
\begin{equation}
f_{1}^{(a_{2})}(\omega )=\pm f_{2}^{(b_{1})}(\omega ),  \label{12}
\end{equation}%
$\nu $-almost everywhere\footnote{%
This terminology means that relation (\ref{12}) can be violated only on a
subset $\Omega ^{\prime }\subset \Omega $ of zero measure $\nu (\Omega
^{\prime })=0.$} (\emph{a.e}$.$) on $\Omega .$ Since, in an arbitrary LHV
model, the values of functions $f_{1}^{(a_{2})},$ $f_{2}^{(b_{1})}$ do not
need to coincide with Alice's and Bob's outcomes under their measurements
specified by settings $a_{2}$ and $b_{1},$ respectively, relation (\ref{12}) 
\emph{does not}, in general, mean the perfect correlation or anticorrelation
of Alice's and Bob's outcomes under the joint measurement $(a_{2},b_{1}).$

Below we prove (propositions 1 - 3) that, for any spectral type of outcomes
and in an LHV model of any type, the LHV condition (\ref{8}) is more general
than the assumption on perfect correlations or anticorrelations and
incorporates the latter assumption only as a particular case.

\begin{proposition}
Let a $2\times 2$ - setting bipartite correlation scenario, with outcomes $%
\lambda _{1},\lambda _{2}=\pm 1,$ admit a correlation LHV model (\ref{7})
and, under the joint measurement $(a_{2},b_{1}),$ Alice's and Bob's outcomes
be perfectly correlated or anticorrelated\footnote{%
For outcomes $\pm 1,$ relations (\ref{5}) and (\ref{13}) are equivalent.}:%
\begin{equation}
\langle \lambda _{1}\lambda _{2}\rangle ^{(a_{2},b_{1})}=\pm 1.  \label{13}
\end{equation}%
Then this LHV model is conditioned by relation (\ref{8}) and, therefore, by
theorem 1, the original Bell inequality (\ref{9}) holds in its "perfect
correlation" (minus sign) or "perfect anticorrelation" (plus sign) form,
respectively.
\end{proposition}

\begin{proof}
In view of Eqs. (\ref{7}), (\ref{13}),%
\begin{equation}
\dint f_{1}^{(a_{2})}(\omega )f_{2}^{(b_{1})}(\omega )\text{ }\nu (\mathrm{d}%
\omega )=\pm 1,  \label{14}
\end{equation}%
where $f_{1}^{(a_{2})}(\omega ),$ $f_{2}^{(b_{1})}(\omega )\in \lbrack
-1,1]. $ For functions with values in $[-1,1]$, the plus sign or the minus
sign version of Eq. (\ref{14}) correspondingly implies:%
\begin{eqnarray}
f_{1}^{(a_{2})}(\omega )f_{2}^{(b_{1})}(\omega ) &=&1\text{ \ \ \ }%
\Leftrightarrow \text{ \ \ \ }f_{1}^{(a_{2})}(\omega
)=f_{2}^{(b_{1})}(\omega )\in \{-1,1\},\text{ \ \ \ }\nu \text{-}a.e.
\label{14-1} \\
&&\text{or}  \notag \\
f_{1}^{(a_{2})}(\omega )f_{2}^{(b_{1})}(\omega ) &=&-1\text{ \ \ }%
\Leftrightarrow \text{ \ \ \ }f_{1}^{(a_{2})}(\omega
)=-f_{2}^{(b_{1})}(\omega )\in \{-1,1\},\text{ \ \ \ }\nu \text{-}a.e. 
\notag
\end{eqnarray}%
These relations mean the validity of condition (\ref{12}) and, therefore,
condition (\ref{8}).
\end{proof}

\bigskip

If Alice's and Bob's outcomes take \emph{any }values in $[-1,1],$ possibly
not discrete, then the assumption on perfect correlations or
anticorrelations under the joint measurement $(a_{2},b_{1})$ is
mathematically expressed by condition (\ref{5}) but not by Bell's
restriction (\ref{13}). For this general case, we have the following
statement (see also proposition 3 below).

\begin{proposition}
Let, for a $2\times 2$ - setting bipartite scenario, with outcomes $\lambda
_{1},\lambda _{2}\in \lbrack -1,1]$ of an arbitrary spectral type, the
averages%
\begin{equation}
\langle \lambda _{1}^{m}\lambda _{2}^{n}\rangle ^{(a_{i},b_{k})},\text{ \ \
\ }m+n\leq 2,\text{ \ \ }m,n=0,1,2,\text{ \ \ }i,k=1,2,  \label{15}
\end{equation}%
admit an LHV\ model 
\begin{equation}
\langle \lambda _{1}^{m}\lambda _{2}^{n}\rangle
^{(a_{i},b_{k})}=\dint\limits_{\Omega }\left( f_{1}^{(a_{i})}(\omega
)\right) ^{m}\left( f_{2}^{(b_{k})}(\omega )\right) ^{n}\text{ }\nu (\mathrm{%
d}\omega )  \label{16}
\end{equation}%
where functions $f_{1}^{(a_{i})}(\omega ),$ $f_{2}^{(b_{k})}(\omega )\in
\lbrack -1,1].$ If, under the joint measurement $(a_{2},b_{1}),$ Alice's and
Bob's outcomes are perfectly correlated or anticorrelated, that is:%
\begin{eqnarray}
P^{(a_{2},b_{1})}(\{\lambda _{1} &=&\lambda _{2}\})=\dint\limits_{\lambda
_{1}=\lambda _{2}}P^{(a_{2},b_{1})}(\mathrm{d}\lambda _{1}\times \mathrm{d}%
\lambda _{2})=1  \label{16-1} \\
&&\text{or}  \notag \\
P^{(a_{2},b_{1})}(\{\lambda _{1} &=&-\lambda _{2}\})=\dint\limits_{\lambda
_{1}=-\lambda _{2}}P^{(a_{2},b_{1})}(\mathrm{d}\lambda _{1}\times \mathrm{d}%
\lambda _{2})=1,  \notag
\end{eqnarray}%
respectively, then the LHV model (\ref{16}) is conditioned by condition (\ref%
{8}) and, therefore, by theorem 1, the original Bell inequality (\ref{9})
holds in its "perfect correlation" (minus sign) or "perfect anticorrelation"
(plus sign) version, respectively.
\end{proposition}

\begin{proof}
Due to representation (\ref{16}) and the assumption on perfect correlations
(the first line of Eq. (\ref{16-1})), we have:%
\begin{eqnarray}
0 &\leq &\dint\limits_{\Omega }(f_{1}^{(a_{2})}(\omega
)-f_{2}^{(b_{1})}(\omega ))^{2}\text{ }\nu (\mathrm{d}\omega )  \label{17} \\
&=&\dint (\lambda _{1}-\lambda _{2})^{2}\text{ }P^{(a_{2},b_{1})}(\mathrm{d}%
\lambda _{1}\times \mathrm{d}\lambda _{2})  \notag \\
&=&\dint\limits_{\lambda _{1}\neq \lambda _{2}}(\lambda _{1}-\lambda
_{2})^{2}P^{(a_{2},b_{1})}(\mathrm{d}\lambda _{1}\times \mathrm{d}\lambda
_{2})  \notag \\
&\leq &4\dint\limits_{\lambda _{1}\neq \lambda _{2}}P^{(a_{2},b_{1})}(%
\mathrm{d}\lambda _{1}\times \mathrm{d}\lambda _{2})  \notag \\
&=&0.  \notag
\end{eqnarray}%
The similar relations (but with the plus sign in the first three lines and
event $\{\lambda _{1}\neq -\lambda _{2}\}$ in the third and the fourth
lines) hold in case of perfect anticorrelations. These relations imply the
validity of condition (\ref{12}), hence, condition (\ref{8}).
\end{proof}

\bigskip

Suppose now that a $2\times 2$ - setting correlation experiment admits an
LHV model\emph{\ }for joint probability distributions - in the sense that
each joint distribution $P^{(a_{i},b_{k})},$ $i,k=1,2,$ admits representation%
\begin{equation}
P^{(a_{i},b_{k})}(\mathrm{d}\lambda _{1}\times \mathrm{d}\lambda
_{2})=\dint\limits_{\Omega }P_{1}^{(a_{i})}(\mathrm{d}\lambda _{1}|\omega
)P_{2}^{(b_{k})}(\mathrm{d}\lambda _{2}|\omega )\text{ }\nu (\mathrm{d}%
\omega ),\text{ \ \ }i,k=1,2,  \label{18}
\end{equation}%
in terms of some variables $\omega \in \Omega $ and conditional probability
distributions $P_{1}^{(a_{i})}(\mathrm{\cdot }|\omega ),$ $P_{2}^{(b_{k})}(%
\mathrm{\cdot }|\omega )$, defined $\nu $-\emph{a.e}$.$ on $\Omega .$ In
this LHV model, expectations $\langle \lambda _{1}\lambda _{2}\rangle
^{(a_{i},b_{k})},$ $i,k=1,2,$ admit representation (\ref{7}) with functions%
\begin{equation}
f_{1}^{(a_{i})}(\omega ):=\dint \lambda _{1}P_{1}^{(a_{i})}(\mathrm{d}%
\lambda _{1}|\omega ),\text{ \ \ }\ f_{2}^{(b_{k})}(\omega ):=\dint \lambda
_{2}P_{2}^{(b_{k})}(\mathrm{d}\lambda _{2}|\omega ),\text{ \ \ }i,k=1,2,
\label{19}
\end{equation}%
so that condition (\ref{8}) takes the form%
\begin{equation}
\dint \lambda _{2}^{^{\prime }}(\lambda _{2}\mp \lambda _{1}^{\prime })\text{
}\mu (\mathrm{d}\lambda _{1}^{\prime }\times \mathrm{d}\lambda _{2}\times 
\mathrm{d}\lambda _{2}^{\prime })\geq 0,  \label{21}
\end{equation}%
where 
\begin{equation}
\mu (\mathrm{d}\lambda _{1}^{\prime }\times \mathrm{d}\lambda _{2}\times 
\mathrm{d}\lambda _{2}^{\prime })=\dint\limits_{\Omega }P_{1}^{(a_{2})}(%
\mathrm{d}\lambda _{1}^{\prime }|\omega )P_{2}^{(b_{1})}(\mathrm{d}\lambda
_{2}|\omega )P_{2}^{(b_{2})}(\mathrm{d}\lambda _{2}^{\prime }|\omega )\text{ 
}\nu (\mathrm{d}\omega ).  \label{22}
\end{equation}

Note that, for outcomes of any spectral type, the existence of an LHV model (%
\ref{7}) for correlation functions does not need to imply the existence of
an LHV model (\ref{18}) for joint probability distributions.

In view of Eqs. (\ref{18}) - (\ref{22}), we have the following corollary of
theorem 1.

\begin{corollary}
Let a $2\times 2$ - setting bipartite correlation scenario, with outcomes $%
\lambda _{1},\lambda _{2}\in \lbrack -1,1]$ of any spectral type, admit a
conditional LHV model (\ref{18}), (\ref{21}) for joint probability
distributions. Then the original Bell inequality (\ref{9}) holds.
\end{corollary}

Let us now show that, as it is the case for correlation LHV models,
discussed above in propositions 1, 2, for an LHV model for joint probability
distributions, condition\ (\ref{21}) incorporates the assumption on perfect
correlations or anticorrelations only as a particular case.

\begin{proposition}
Let a $2\times 2$ - setting bipartite correlation scenario, with outcomes $%
\lambda _{1},\lambda _{2}\in \Lambda \subseteq \lbrack -1,1]$ of an
arbitrary spectral type, admit an LHV model (\ref{18}) for joint probability
distributions. If, under the joint measurement $(a_{2},b_{1}),$ Alice's and
Bob's outcomes are perfectly correlated or anticorrelated, i. e. assumption (%
\ref{16-1}) is fulfilled, then this LHV\ model is conditioned by the minus
sign or plus sign version of condition (\ref{21}), respectively, and,
therefore, by corollary 1, the original Bell inequality (\ref{9}) holds.
\end{proposition}

\begin{proof}
From Eqs. (\ref{21}), (\ref{22}) and the assumption on perfect correlations
(the first line of Eq. (\ref{16-1})) it follows:%
\begin{eqnarray}
0 &\leq &\left\vert \dint \lambda _{2}^{^{\prime }}(\lambda _{2}-\lambda
_{1}^{\prime })\text{ }\mu (\mathrm{d}\lambda _{1}^{\prime }\times \mathrm{d}%
\lambda _{2}\times \mathrm{d}\lambda _{2}^{\prime })\right\vert  \label{23}
\\
&\leq &\dint \left\vert \lambda _{1}^{\prime }-\lambda _{2}\right\vert \text{
}\mu (\mathrm{d}\lambda _{1}^{\prime }\times \mathrm{d}\lambda _{2}\times
\Lambda )  \notag \\
&=&\dint\limits_{\lambda _{1}^{\prime }\neq \lambda _{2}}\left\vert \lambda
_{1}^{\prime }-\lambda _{2}\right\vert \text{ }P^{(a_{2},b_{1})}(\mathrm{d}%
\lambda _{1}^{\prime }\times \mathrm{d}\lambda _{2})  \notag \\
&\leq &2\dint\limits_{\lambda _{1}^{\prime }\neq \lambda
_{2}}P^{(a_{2},b_{1})}(\mathrm{d}\lambda _{1}^{\prime }\times \mathrm{d}%
\lambda _{2})  \notag \\
&=&0.  \notag
\end{eqnarray}%
These relations imply the validity of the minus sign version of condition (%
\ref{21}). The similar relations (but with the plus sign in the first three
lines and event $\{\lambda _{1}^{\prime }\neq -\lambda _{2}\}$ in the third
and fourth lines) hold in case of perfect anticorrelations. This proves the
statement.
\end{proof}

\bigskip

From propositions 1, 2, 3 it follows that, in an LHV model of any type, the
assumption on perfect correlations or anticorrelations \emph{implies} the
validity of the general LHV condition (\ref{8}). The converse of this
statement is not, however, true.

In order to show this, let us consider a \emph{dichotomic }bipartite case%
\footnote{%
In a dichotomic case, the existence of an LHV model for correlation
functions implies the existence of an LHV model for joint probability
distributions, see theorem 3 in section 4.1 of \cite{9}.} with outcomes $\pm
1.$ Since $(\lambda _{1}^{\prime })^{2}=1,$ we have the following expression
for the left-hand side of, for example, the minus sign form of condition (%
\ref{21}):%
\begin{eqnarray}
&&\dint \lambda _{2}^{^{\prime }}(\lambda _{2}-\lambda _{1}^{\prime })\text{ 
}\mu (\mathrm{d}\lambda _{1}^{\prime }\times \mathrm{d}\lambda _{2}\times 
\mathrm{d}\lambda _{2}^{\prime })  \label{cv} \\
&=&\dint \lambda _{1}^{\prime }\lambda _{2}^{^{\prime }}(\lambda
_{1}^{\prime }\lambda _{2}-1)\text{ }\mu (\mathrm{d}\lambda _{1}^{\prime
}\times \mathrm{d}\lambda _{2}\times \mathrm{d}\lambda _{2}^{\prime }) 
\notag \\
&=&2\mu (\{\lambda _{1}^{\prime }\lambda _{2}=-1\})-4\mu (\{\lambda
_{1}^{\prime }\lambda _{2}=-1,\text{ }\lambda _{1}^{\prime }\lambda
_{2}^{\prime }=1\}).  \notag
\end{eqnarray}%
Taking into account that 
\begin{eqnarray}
\mu (\{\lambda _{1}^{\prime }\lambda _{2} &=&-1,\text{ }\lambda _{1}^{\prime
}\lambda _{2}^{\prime }=1\})\leq \mu (\{\lambda _{1}^{\prime }\lambda
_{2}=-1\}),  \label{sb} \\
\mu (\{\lambda _{1}^{\prime }\lambda _{2} &=&-1,\text{ }\lambda _{1}^{\prime
}\lambda _{2}^{\prime }=1\})\leq \mu (\{\lambda _{1}^{\prime }\lambda
_{2}^{\prime }=1\}),  \notag \\
\mu (\{\lambda _{1}^{\prime }\lambda _{2} &=&-1,\text{ }\lambda _{1}^{\prime
}\lambda _{2}^{\prime }=1\})\leq \sqrt{\mu (\{\lambda _{1}^{\prime }\lambda
_{2}=-1\})\text{ }\mu (\{\lambda _{1}^{\prime }\lambda _{2}^{\prime }=1\})},
\notag
\end{eqnarray}%
and substituting these relations into Eq. (\ref{cv}), we derive:%
\begin{eqnarray}
&&\dint \lambda _{2}^{^{\prime }}(\lambda _{2}-\lambda _{1}^{\prime })\text{ 
}\mu (\mathrm{d}\lambda _{1}^{\prime }\times \mathrm{d}\lambda _{2}\times 
\mathrm{d}\lambda _{2}^{\prime })  \label{w} \\
&\geq &2\mu (\{\lambda _{1}^{\prime }\lambda _{2}=-1\})-4\mu (\{\lambda
_{1}^{\prime }\lambda _{2}^{\prime }=1\})  \notag
\end{eqnarray}%
and 
\begin{eqnarray}
&&\dint \lambda _{2}^{^{\prime }}(\lambda _{2}-\lambda _{1}^{\prime })\text{ 
}\mu (\mathrm{d}\lambda _{1}^{\prime }\times \mathrm{d}\lambda _{2}\times 
\mathrm{d}\lambda _{2}^{\prime })  \label{www} \\
&\geq &2\sqrt{\mu (\{\lambda _{1}^{\prime }\lambda _{2}=-1\})}  \notag \\
&&\times (\sqrt{\mu (\{\lambda _{1}^{\prime }\lambda _{2}=-1\})}-2\sqrt{\mu
(\{\lambda _{1}^{\prime }\lambda _{2}^{\prime }=1\})}\text{ }).  \notag
\end{eqnarray}%
Therefore, if either of conditions 
\begin{equation}
\mu (\{\lambda _{1}^{\prime }\lambda _{2}=-1\})\geq 2\mu (\{\lambda
_{1}^{\prime }\lambda _{2}^{\prime }=1\})\text{ \ \ or \ \ }\mu (\{\lambda
_{1}^{\prime }\lambda _{2}=-1\})=0  \label{1w}
\end{equation}%
is fulfilled, then the minus sign version of the general LHV condition (\ref%
{21}) holds. Since from Eq (\ref{22}) it follows that 
\begin{eqnarray}
\mu (\{\lambda _{1}^{\prime }\lambda _{2} &=&-1\})\equiv
P^{(a_{2},b_{1})}(\{\lambda _{1}^{\prime }\lambda _{2}=-1\}), \\
\mu (\{\lambda _{1}^{\prime }\lambda _{2}^{\prime } &=&1\})\equiv
P^{(a_{2},b_{2})}(\{\lambda _{1}^{\prime }\lambda _{2}^{\prime }=1\}), 
\notag
\end{eqnarray}%
in terms of joint probabilities, conditions (\ref{1w}) read:%
\begin{eqnarray}
P^{(a_{2},b_{1})}(\{\lambda _{1}\lambda _{2}
&=&1\})+2P^{(a_{2},b_{2})}(\{\lambda _{1}\lambda _{2}=1\})\leq 1  \label{ww}
\\
&&\text{or \ \ }  \notag \\
P^{(a_{2},b_{1})}(\{\lambda _{1}\lambda _{2} &=&1\})=1,  \notag
\end{eqnarray}%
respectively.

Thus, if, under the joint measurements $(a_{2},b_{1})$ and $(a_{2},b_{2}),$
the probabilities of event $\{\lambda _{1}\lambda _{2}=1\}$ satisfy \emph{%
either} of conditions in (\ref{ww}), then the minus sign form of the LHV
condition (\ref{21}) is fulfilled so that, according to corollary 1, the
"perfect correlation" version 
\begin{equation}
\left\vert \langle \lambda _{1}\lambda _{2}\rangle ^{(a_{1},b_{1})}-\langle
\lambda _{1}\lambda _{2}\rangle ^{(a_{1},b_{2})}\right\vert \leq 1-\langle
\lambda _{1}\lambda _{2}\rangle ^{(a_{2},b_{2})}
\end{equation}%
of the original Bell inequality (\ref{9}) holds.

Quite similarly, if, under the joint measurements $(a_{2},b_{1})$ and $%
(a_{2},b_{2}),$ the probabilities of event $\{\lambda _{1}\lambda _{2}=-1\}$
satisfy either of conditions%
\begin{eqnarray}
P^{(a_{2},b_{1})}(\{\lambda _{1}\lambda _{2}
&=&-1\})+2P^{(a_{2},b_{2})}(\{\lambda _{1}\lambda _{2}=-1\})\leq 1
\label{ww'} \\
&&\text{or \ \ }  \notag \\
P^{(a_{2},b_{1})}(\{\lambda _{1}\lambda _{2} &=&-1\})=1,  \notag
\end{eqnarray}%
then the "perfect anticorrelation" version 
\begin{equation}
\left\vert \langle \lambda _{1}\lambda _{2}\rangle ^{(a_{1},b_{1})}-\langle
\lambda _{1}\lambda _{2}\rangle ^{(a_{1},b_{2})}\right\vert \leq 1+\langle
\lambda _{1}\lambda _{2}\rangle ^{(a_{2},b_{2})}
\end{equation}%
of the original Bell inequality (\ref{9}) holds.

We stress that \emph{all conditions in }(\ref{ww}), (\ref{ww'})\emph{\ can
be tested experimentally and that only the second condition in }(\ref{ww})%
\emph{\ and the second condition in }(\ref{ww'})\emph{\ mean,
correspondingly, perfect correlations and perfect anticorrelations under the
joint measurement }$(a_{2},b_{1}).$

In the following sections, we introduce classical and quantum correlation
scenarios where the general LHV condition (\ref{8}) and, therefore, the
original Bell inequality (\ref{9}) are fulfilled for the whole range of
measurement settings.

\section{Classical bipartite case}

As an application of our results derived in section 2, consider the
probabilistic description of a $2\times 2$ - setting bipartite correlation
scenario, where every joint measurement $(a_{i},b_{k}),$ $i,k=1,2,$ is
performed on the same, identically prepared pair of classical particles,
each observed at one of sites and not interacting with each other during a
joint measurement. Let also a measurement device of each party do not affect
a measurement device and a particle observed at another site, and results of
each joint measurement do not in any way disturb results of other joint
measurements. Due to this physical setting, the considered classical
correlation scenario is \emph{local} in the Einstein-Podolsky-Rosen (EPR) 
\cite{12} sense\footnote{%
On the mathematical formulation of \emph{the EPR locality}, see section 3 in 
\cite{9} and also section 4 in \cite{13}.}.

In a classical EPR local bipartite case\footnote{%
See section 3.1 in \cite{9}.}, a state of a bipartite classical system
before measurements is represented by a probability distribution $\pi $ of
some \emph{system's} \emph{variables} $\theta \in \Theta $ such that,\emph{\ 
}for any joint measurement $(a_{i},b_{k})$ performed on this classical
system in a state $\pi ,$ distribution $P_{_{\pi }}^{(a_{i},b_{k})}$ has the
factorizable form%
\begin{equation}
P_{\pi }^{(a_{i},b_{k})}(\mathrm{d}\lambda _{1}\times \mathrm{d}\lambda
_{2})=\dint\limits_{\Theta }P_{1}^{(a_{i})}(\mathrm{d}\lambda _{1}|\theta
)P_{2}^{(b_{k})}(\mathrm{d}\lambda _{2}|\theta )\pi (\mathrm{d}\theta ),
\label{23.2}
\end{equation}%
reducing to an image (\ref{23.6}) of distribution $\pi $ if Alice's and
Bob's classical measurements are \emph{ideal,} i.e. describe a measured
system property without an error. Due to the EPR locality of the considered
classical correlation experiment, each of conditional distributions $%
P_{1}^{(a_{i})}(\mathrm{\cdot }|\theta ),$ $P_{2}^{(b_{k})}(\mathrm{\cdot }%
|\theta )$ depends only on a setting of the corresponding measurement at the
corresponding site.

Substituting representation (\ref{23.2}) into Eqs. (\ref{3}), (\ref{4}), we
get the following expressions for averages of Alice's $(n=1)$ and Bob's $%
(n=2)$ outcomes:%
\begin{eqnarray}
\langle \lambda _{1}\rangle _{\pi }^{(a_{i},b_{1})} &=&\langle \lambda
_{1}\rangle _{\pi }^{(a_{i},b_{2})}=\dint\limits_{\Theta }A_{i}(\theta )\pi (%
\mathrm{d}\theta ):=\langle \lambda _{1}\rangle _{\pi }^{(A_{i})},\text{ \ \ 
}i=1,2,  \label{23.3} \\
\langle \lambda _{2}\rangle _{\pi }^{(a_{1},b_{k})} &=&\langle \lambda
_{2}\rangle _{\pi }^{(a_{2},b_{k})}=\dint\limits_{\Theta }B_{k}(\theta )\pi (%
\mathrm{d}\theta ):=\langle \lambda _{2}\rangle _{\pi }^{(B_{k})},\text{ \ \ 
}k=1,2,  \notag
\end{eqnarray}%
and the product expectations 
\begin{equation}
\langle \lambda _{1}\lambda _{2}\rangle _{\pi
}^{(a_{i},b_{k})}=\dint\limits_{\Theta }A_{i}(\theta )B_{k}(\theta )\pi (%
\mathrm{d}\theta ):=\langle \lambda _{1}\lambda _{2}\rangle _{\pi
}^{(A_{i},B_{k})},\text{ \ \ }i,k=1,2,  \label{23.4}
\end{equation}%
in terms of classical observables%
\begin{eqnarray}
A_{i}(\theta ) &:&=\langle \lambda _{1}|\theta \rangle _{\pi
}^{(a_{i})}=\dint \lambda _{1}P_{1}^{(a_{i})}(\mathrm{d}\lambda _{1}|\theta
)\in \lbrack -1,1],  \label{23.5} \\
B_{k}(\theta ) &:&=\langle \lambda _{2}|\theta \rangle _{\pi
}^{(b_{k})}=\dint \lambda _{2}P_{2}^{(b_{k})}(\mathrm{d}\lambda _{2}|\theta
)\in \lbrack -1,1].  \notag
\end{eqnarray}

If classical measurements, specified by settings $a_{i}$ and $b_{k}$, are 
\emph{ideal}, then values of observables $A_{i}$ and $B_{k}$ coincide with
Alice's and Bob's outcomes under these measurements, while the joint
distribution $P_{\pi ,\text{ideal}}^{(a_{i},b_{k})}$ has the image form 
\begin{equation}
P_{\pi ,\text{ideal}}^{(a_{i},b_{k})}(D_{1}\times D_{2})=\pi
(A_{i}^{-1}(D_{1})\cap B_{k}^{-1}(D_{2})),\ \ \ \forall D_{1},D_{2}\subseteq
\lbrack -1,1].  \label{23.6}
\end{equation}%
Here, $A^{-1}(D):=\{\theta \in \Theta \mid A(\theta )\in D\}$ denotes the
preimage of a subset $D\subseteq \lbrack -1,1]$ under mapping $A:\Theta
\rightarrow \lbrack -1,1].$ Similarly, for notation $B^{-1}(D).$

If classical measurements $a_{i}$ and $b_{k}$ are \emph{non-ideal}\footnote{%
A non-ideal classical measurement is called \emph{randomized}.}, then
observables $A_{i}$ and $B_{k}$ describe these measurements only \emph{in
average }-\emph{\ }in the sense that values $A_{i}(\theta ),$ $B_{k}(\theta
) $ of these observables constitute not Alice's and Bob's outcomes but only
their averages $\langle \lambda _{1}|\theta \rangle _{\pi }^{(a_{i})},$ $%
\langle \lambda _{2}|\theta \rangle _{\pi }^{(b_{k})}$ for a certain initial 
$\theta \in \Theta .$

Thus, the statistical context of our notation $\langle \lambda _{1}\lambda
_{2}\rangle _{\pi }^{(A_{i},B_{k})}$\emph{\ is, in general, different }from
the context of notation $E_{\pi }(A_{i}B_{k})$ used in conventional
probability theory for the expected value of the product of random variables.

From representations (\ref{23.2}), (\ref{23.4}) it follows that any EPR
local $2\times 2$ - setting correlation scenario on a classical state $\pi $
admits the LHV model (\ref{18}) for joint probability distributions and,
hence, the LHV model (\ref{7}) for correlation functions. Therefore, if
classical observables (\ref{23.5}) satisfy condition (\ref{8}), then by
theorem 1 the original Bell inequality (\ref{9}) holds and, in view of
notation (\ref{23.4}), reads:%
\begin{equation}
\left\vert \langle \lambda _{1}\lambda _{2}\rangle _{\pi
}^{(A_{1},B_{1})}-\langle \lambda _{1}\lambda _{2}\rangle _{\pi
}^{(A_{1},B_{2})}\right\vert \leq 1\mp \langle \lambda _{1}\lambda
_{2}\rangle _{\pi }^{(A_{2},B_{2})}.  \label{23.7}
\end{equation}

If classical observables corresponding to settings $a_{2}$ and $b_{1}$
coincide:$A_{2}=B_{1},$ then the minus sign form of condition (\ref{8})
holds. Therefore, an arbitrary classical state $\pi $ satisfies the "perfect
correlation" version%
\begin{equation}
\left\vert \langle \lambda _{1}\lambda _{2}\rangle _{\pi
}^{(A_{1},B_{1})}-\langle \lambda _{1}\lambda _{2}\rangle _{\pi
}^{(A_{1},B_{2})}\right\vert \leq 1-\langle \lambda _{1}\lambda _{2}\rangle
_{\pi }^{(B_{1},B_{2})}  \label{23.8}
\end{equation}%
of the original Bell inequality (\ref{9}) for any three classical
observables $A_{1},$ $A_{2}=B_{1},$ $B_{2}$, measured, possibly \emph{in
average}, at sites of Alice and Bob and having values in $[-1,1]$.

If a joint classical measurement $(a_{2},b_{1})$ on a state $\pi $ is ideal
and the corresponding classical observables coincide: $A_{2}=B_{1}$, then,
due to Eq. (\ref{23.6}), the outcome event $\{\lambda _{1}=\lambda _{2}\}$
is observed with certainty%
\begin{equation}
P_{\pi ,ideal}^{(B_{1},B_{1})}(\{\lambda _{1}=\lambda _{2}\}):=P_{\pi
,ideal}^{(a_{2},b_{1})}(\{\lambda _{1}=\lambda _{2}\})|_{_{A_{2}=B_{1}}}=1.
\label{23.9}
\end{equation}%
Hence, under an \emph{ideal} joint classical measurement of the same
classical observable at both sites, Alice's and Bob's outcomes are
necessarily perfectly correlated for any initial classical state $\pi .$

However, if classical measurements of Alice and Bob on a state $\pi $ are 
\emph{non-ideal}, then a classical state $\pi $ \emph{does not need }to
exhibit\emph{\ }perfect correlations whenever the same classical observable
is \emph{in average} measured at both sites.

Thus, the "perfect correlation" version (\ref{23.8}) of the original Bell
inequality is fulfilled for any classical state $\pi $ and any three
classical observables $A_{1},$ $A_{2}=B_{1},$ $B_{2}$ measured ideally or
non-ideally at sites of Alice and Bob but the condition on perfect
correlations necessarily holds only in case of ideal classical measurements
of observable $B_{1}$ at both sites.

Consider now a possible \emph{physical context }of a correlation scenario on
a classical state $\pi $ modeled at Alice's and Bob's sites by only three
classical observables $A_{1},$ $A_{2}=B_{1},$ $B_{2}.$ Let Alice/Bob joint
measurements be performed by identical measurement devices with settings $%
a_{1},$ $a_{2}=b_{1},$ $b_{2}$ and upon a pair of classical particles, that
are indistinguishable, identically prepared and do not interact with each
other during measurements. In this case, due to the \emph{physical
indistinguishability} of sites and measurements of Alice and Bob specified
by the same setting $a_{2}=b_{1},$ these measurements are described by the
same conditional averages $\langle \lambda _{1}|\theta \rangle _{\pi
}^{(a_{2}=b_{1})}=\langle \lambda _{2}|\theta \rangle _{\pi }^{(b_{1})}$ for
any initial $\theta \in \Theta ,$ and, therefore, in view of relations (\ref%
{23.5}) should be in average modeled by the same classical observable $%
A_{2}=B_{1}$ at both sites.

\section{Quantum bipartite case}

Consider now the probabilistic description of a quantum EPR$\ $local $%
2\times 2$ - setting correlation scenario, where any of joint measurements $%
(a_{i},b_{k}),$ $i,k=1,2$, is performed on the same identically prepared
pair of \emph{quantum} particles and the physical context of this quantum
scenario is similar to that of a \emph{classical} EPR$\ $local scenario
discussed in section 3 - with the only substitution of term "classical
particles" by term "quantum particles".

In a quantum EPR local bipartite case\footnote{%
See section 3.1 in \cite{9}.}, a state of a bipartite quantum system before
measurements is represented by a density operator $\rho $ on a complex
separable Hilbert space $\mathcal{H}_{1}\otimes \mathcal{H}_{2}$, possibly
infinitely dimensional, and, for any joint measurement $(a_{i},b_{k})$
performed on this quantum system in a state $\rho $, the joint probability
distribution $P_{\rho }^{(a_{i},b_{k})}$ takes the form%
\begin{equation}
P_{\rho }^{(a_{i},b_{k})}(\mathrm{d}\lambda _{1}\times \mathrm{d}\lambda
_{2})=\mathrm{tr}[\rho \{\mathrm{M}_{1}^{(a_{i})}(\mathrm{d}\lambda
_{1})\otimes \mathrm{M}_{2}^{(b_{k})}(\mathrm{d}\lambda _{2})\}],  \label{24}
\end{equation}%
where $\mathrm{M}_{1}^{(a_{i})},$ $\mathrm{M}_{2}^{(b_{k})}$ are positive
operator-valued (POV) measures, representing on $\mathcal{H}_{1}$ and $%
\mathcal{H}_{2}$ the corresponding quantum measurements of Alice and Bob.
Due to the EPR locality of the considered quantum correlation experiment,
each of these POV measures depends only on a setting of the corresponding
measurement at the corresponding site.

Substituting Eq. (\ref{24}) into Eqs. (\ref{3}), (\ref{4}), we get the
following expressions for averages of Alice's $(n=1)$ and Bob's $(n=2)$
outcomes: 
\begin{eqnarray}
\langle \lambda _{1}\rangle _{\rho }^{(a_{i},b_{1})} &=&\langle \lambda
_{1}\rangle _{\rho }^{(a_{i},b_{2})}=\mathrm{tr}[\rho \{A_{i}\otimes \mathbb{%
I}_{\mathcal{H}_{2}}\mathbb{\}]}:=\langle \lambda _{1}\rangle _{\rho
}^{(A_{i})},\mathbb{\ \ \ \ }i=1,2,  \label{24-1} \\
&&  \notag \\
\langle \lambda _{2}\rangle _{\rho }^{(a_{1},b_{k})} &=&\langle \lambda
_{2}\rangle _{\rho }^{(a_{2},b_{k})}=\mathrm{tr}[\rho \{\mathbb{I}_{\mathcal{%
H}_{1}}\mathbb{\otimes }B_{k}\mathbb{\}]}:=\langle \lambda _{2}\rangle
_{\rho }^{(B_{k})},\mathbb{\ \ \ \ }k=1,2,  \notag
\end{eqnarray}%
and the product expectations%
\begin{eqnarray}
\langle \lambda _{1}\lambda _{2}\rangle _{\rho }^{(a_{i},b_{k})} &=&\dint
\lambda _{1}\lambda _{2}\mathrm{tr}[\rho \{\mathrm{M}_{1}^{(a_{i})}(\mathrm{d%
}\lambda _{1})\otimes \mathrm{M}_{2}^{(b_{k})}(\mathrm{d}\lambda _{2})\}]
\label{25} \\
&=&\mathrm{tr}[\rho \{A_{i}\otimes B_{k}\}]:=\langle \lambda _{1}\lambda
_{2}\rangle _{\rho }^{(A_{i},B_{k})},\ \ \ \ i,k=1,2,  \notag
\end{eqnarray}%
in terms of quantum observables%
\begin{equation}
A_{i}:=\dint \lambda _{1}\mathrm{M}_{1}^{(a_{i})}(\mathrm{d}\lambda _{1}),%
\text{ \ \ \ }B_{k}:=\dint \lambda _{2}\mathrm{M}_{2}^{(b_{k})}(\mathrm{d}%
\lambda _{2})  \label{26}
\end{equation}%
with eigenvalues in $[-1,1]$, on $\mathcal{H}_{1}$ and $\mathcal{H}_{2},$
respectively.

If Alice's and Bob's quantum measurements specified by settings $a_{i}$ and $%
b_{k}$ are \emph{ideal}, then eigenvalues of observables $A_{i}$ and $B_{k}$
coincide with Alice's and Bob's outcomes under these measurements while the
POV measures $\mathrm{M}_{1}^{(a_{i})},$ $\mathrm{M}_{2}^{(b_{k})}$
describing these measurements are given by projection-valued\footnote{%
Due to this, an ideal quantum measurement is often referred to as projective.%
} measures $\mathrm{P}_{A_{i}}$, $\mathrm{P}_{B_{k}},$ uniquely
corresponding to quantum observables $A_{i}$, $B_{k}$ due to the spectral
theorem. If quantum measurements $a_{i}$ and $b_{k}$ are \emph{non-ideal}%
\footnote{%
In a quantum case, an non-ideal measurement is otherwise referred to as 
\emph{generalized}.}, then observables (\ref{26}) describe these quantum
measurements only \emph{in average} - in the sense of representations (\ref%
{24-1}), (\ref{25}).

For Alice/Bob joint measurements on a bipartite quantum state $\rho $ on $%
\mathcal{H}_{1}\otimes \mathcal{H}_{2},$ consider a possibility of a
conditional LHV simulation specified in section 2.

We recall \cite{6, 7} that, for any state $\rho $ on $\mathcal{H}_{1}\otimes 
\mathcal{H}_{2},$ there exist self-adjoint trace class operators $%
T_{\blacktriangleright }$ on $\mathcal{H}_{1}\otimes \mathcal{H}_{2}\otimes 
\mathcal{H}_{2}$ and $T_{\blacktriangleleft }$ on $\mathcal{H}_{1}\otimes 
\mathcal{H}_{1}\otimes \mathcal{H}_{2},$ not necessarily positive, such that 
\begin{equation}
\mathrm{tr}_{\mathcal{H}_{2}}^{(2)}[T_{\blacktriangleright }]=\mathrm{tr}_{%
\mathcal{H}_{2}}^{(3)}[T_{\blacktriangleright }]=\rho ,\text{ \ \ \ }\mathrm{%
tr}_{\mathcal{H}_{1}}^{(1)}[T_{\blacktriangleleft }]=\mathrm{tr}_{\mathcal{H}%
_{1}}^{(2)}[T_{\blacktriangleleft }]=\rho .  \label{29}
\end{equation}%
Here, the below indices of $T$ point to a direction of extension of a
Hilbert space $\mathcal{H}_{1}\otimes \mathcal{H}_{2},$ and notation $%
\mathrm{tr}_{\mathcal{H}_{m}}^{(k)}[\cdot ]$ means the partial trace over
the elements of a Hilbert space $\mathcal{H}_{m}$, $m=1,$ $2,$ standing in
the $k$-th place of tensor products. In \cite{6, 7}, we\ refer to any of
these dilations as \emph{a} \emph{source operator} for a bipartite state $%
\rho .$ For each source operator, $\mathrm{tr}[T]=1.$

Introduce also the following new notion. We call a bounded operator $Z\neq 0$
on $\mathcal{H}_{1}\otimes ...\otimes \mathcal{H}_{N}$ as \emph{%
tensor-positive\ }and denote it by $Z\overset{\otimes }{>}0$ if the scalar
product%
\begin{equation}
(\psi _{1}\otimes ...\otimes \psi _{N},\text{ }Z\psi _{1}\otimes ...\otimes
\psi _{N})\geq 0,  \label{30}
\end{equation}%
for any vectors $\psi _{1}\in \mathcal{H}_{1},...,\psi _{N}\in \mathcal{H}%
_{N}.$ Any positive operator is, of course, tensor-positive but the converse
is not true. Due to Eqs. (\ref{30}) and the spectral decomposition of a
positive\footnote{%
On a complex separable Hilbert space, every positive operator is
self-adjoint.} operator on a complex separable Hilbert space, the relation 
\begin{equation}
\mathrm{tr}[Z\{W_{1}\otimes ...\otimes W_{N}\}]\geq 0  \label{30-1}
\end{equation}%
holds for any tensor-positive $Z\overset{\otimes }{>}0$ and any non-negative
bounded operators $W_{n}\geq 0,$ $n=1,..,N,$ each defined on the
corresponding complex separable Hilbert space $\mathcal{H}_{n}.$

The following statement (proved in appendix) specifies the general LHV
condition (\ref{8}) in quantum terms.

\begin{theorem}
Let a quantum $2\times 2$ - setting bipartite correlation scenario, with
outcomes $\lambda _{1},\lambda _{2}\in \Lambda \subseteq \lbrack -1,1]$ of
any spectral type, be specified by Eqs. (\ref{24}) - (\ref{26}) and
performed on a quantum state $\rho $ on $\mathcal{H}_{1}\otimes \mathcal{H}%
_{2}$, that has a tensor-positive source operator $R_{\blacktriangleright }%
\overset{\otimes }{>}0$ on $\mathcal{H}_{1}\otimes \mathcal{H}_{2}\otimes 
\mathcal{H}_{2}$. Then this quantum correlation experiment admits an LHV
model (\ref{18}) for joint probability distributions. If, in addition, 
\begin{equation}
\mathrm{tr}[\sigma _{R_{\blacktriangleright }}\{B_{1}\otimes B_{2}\}]\mp 
\mathrm{tr[}\rho \{A_{2}\otimes B_{2}\}]\geq 0  \label{32-1}
\end{equation}%
where 
\begin{equation}
\sigma _{R_{\blacktriangleright }}:=\mathrm{tr}_{\mathcal{H}%
_{1}}^{(1)}[R_{\blacktriangleright }]  \label{x}
\end{equation}%
is a density operators on $\mathcal{H}_{2}\otimes \mathcal{H}_{2}$, then
this quantum correlation experiment admits the conditional LHV model (\ref%
{18}), (\ref{21}) and, therefore, by corollary 1 satisfies the original Bell
inequality%
\begin{equation}
\left\vert \langle \lambda _{1}\lambda _{2}\rangle _{\rho
}^{(A_{1},B_{1})}-\langle \lambda _{1}\lambda _{2}\rangle _{\rho
}^{(A_{1},B_{2})}\right\vert \leq 1\mp \langle \lambda _{1}\lambda
_{2}\rangle _{\rho }^{(A_{2},B_{2})},  \label{28-1}
\end{equation}%
in its "perfect correlation" (minus sign) or "perfect anticorrelation" (plus
sign) form, respectively. In quantum terms, this inequality reads%
\begin{equation}
\left\vert \text{ }\mathrm{tr}[\rho \{A_{1}\otimes B_{1}\}]-\mathrm{tr}[\rho
\{A_{1}\otimes B_{2}\}]\right\vert \leq 1\mp \mathrm{tr}[\rho \{A_{2}\otimes
B_{2}\}].  \label{28}
\end{equation}%
\bigskip
\end{theorem}

A similar statement holds for a state $\rho $ that has a tensor-positive
source operator $R_{\blacktriangleleft }\overset{\otimes }{>}0$ on $\mathcal{%
H}_{1}\otimes \mathcal{H}_{1}\otimes \mathcal{H}_{2}.$

We stress that condition (53) introduced in \cite{5} for a separable quantum
case and condition (42) in \cite{7} represent particular cases of the
quantum LHV condition (\ref{32-1}).

From theorem 2 it follows that, for a state $\rho $ on $\mathcal{H}\otimes 
\mathcal{H}$ with a tensor-positive source operator $R_{\blacktriangleright }%
\overset{\otimes }{>}0$ and three \emph{given} quantum observables $A_{1},$ $%
A_{2}=B_{1}$, $B_{2}$ on $\mathcal{H}$, the original Bell inequality 
\begin{equation}
\left\vert \text{ }\mathrm{tr}[\rho \{A_{1}\otimes B_{1}\}]-\mathrm{tr}[\rho
\{A_{1}\otimes B_{2}\}]\right\vert \leq 1\mp \mathrm{tr}[\rho \{B_{1}\otimes
B_{2}\}],  \label{37-1}
\end{equation}%
in its minus sign or plus sign version, holds if%
\begin{equation}
\mathrm{tr}[\sigma _{R_{\blacktriangleright }}\{B_{1}\otimes B_{2}\}]\mp 
\mathrm{tr}[\rho \{B_{1}\otimes B_{2}\}]\geq 0,  \label{37-2}
\end{equation}%
respectively. We stress that condition (\ref{37-2}), which is the quantum
version of the general LHV condition (\ref{8}), does not mean the perfect
correlation or anticorrelation of Alice's and Bob's outcomes if observable $%
B_{1}$ is measured at both sites.

The following statement specifies the property of a bipartite quantum state
ensuring the validity of the "perfect correlation" (minus sign) version of
the original Bell inequality (\ref{37-1}) for any three quantum observables $%
A_{1},$ $A_{2}=B_{1},$ $B_{2},$ measured, possibly in average, at sites of
Alice and Bob, respectively.

\begin{theorem}
If, for a quantum state $\rho $ on $\mathcal{H}\otimes \mathcal{H},$ there
exists a tensor-positive source operator $R\overset{\otimes }{>}0$ on $%
\mathcal{H}\otimes \mathcal{H}\otimes \mathcal{H}$ such that 
\begin{equation}
\mathrm{tr}_{\mathcal{H}}^{(k)}[R]=\rho ,\text{ \ \ \ }k=1,2,3,  \label{39}
\end{equation}%
then this state $\rho $ satisfies the "perfect correlation" version of the
original Bell inequality: 
\begin{equation}
\left\vert \langle \lambda _{1}\lambda _{2}\rangle _{\rho
}^{(A_{1},B_{1})}-\langle \lambda _{1}\lambda _{2}\rangle _{\rho
}^{(A_{1},B_{2})}\right\vert \leq 1-\langle \lambda _{1}\lambda _{2}\rangle
_{\rho }^{(B_{1},B_{2})},  \label{40-1}
\end{equation}%
for any three bounded quantum observables $A_{1},$ $A_{2}=B_{1},$ $B_{2}$ on 
$\mathcal{H}$, with eigenvalues in $[-1,1]$ of an arbitrary spectral type,
discrete or continuous. In quantum terms, this inequality reads:%
\begin{equation}
\left\vert \text{ }\mathrm{tr}[\rho \{A_{1}\otimes B_{1}\}]-\mathrm{tr}[\rho
\{A_{1}\otimes B_{2}\}]\right\vert \leq 1-\mathrm{tr}[\rho \{B_{1}\otimes
B_{2}\}].  \label{40}
\end{equation}
\end{theorem}

\begin{proof}
For a source operator $R\overset{\otimes }{>}0$ specified by property (\ref%
{39}), the reduced operator $\sigma _{R},$ given by (\ref{x}), is equal to $%
\rho .$ Therefore, the minus sign version of condition (\ref{32-1}) takes
the form:%
\begin{equation}
\mathrm{tr}[\rho \{B_{1}\otimes B_{2}\}]-\mathrm{tr}[\rho \{A_{2}\otimes
B_{2}\}]\geq 0,  \label{41}
\end{equation}%
which is always true if $A_{2}=B_{1}.$ By theorem 2, this proves the
statement.
\end{proof}

Due to theorems 2, 3, every bipartite state $\rho $ with with the state
property (\ref{39}) (i) admits an LHV description under any bipartite
quantum measurements, ideal or non-ideal, with two settings per site; (ii)
does not need to exhibit perfect correlations (may even have a negative
correlation function $\langle \lambda _{1}\lambda _{2}\rangle _{\rho
}^{(B_{1},\text{ }B_{1})}$) if the same quantum observable is measured at
both sites -- but satisfies the "perfect correlation" version (\ref{40}) of
the original Bell inequality for any three quantum observables $A_{1},$ $%
A_{2}=B_{1},$ $B_{2},$ measured ideally or non-ideally at sites of Alice and
Bob, respectively, and having eigenvalues in $[-1,1].$

Note that an arbitrary separable quantum state does not need to have the
state property (\ref{39}) and, therefore, to satisfy inequality (\ref{40})
for arbitrary $A_{1},$ $A_{2}=B_{1},$ $B_{2}$, and that the class of
bipartite quantum states specified by property (\ref{39}) includes separable
and nonseparable states introduced by us earlier in [5 - 7] only as a
particular subclass.

We stress that the physical context of a quantum correlation scenario
described by three quantum observables $A_{1},$ $A_{2}=B_{1},$ $B_{2}$ is 
\emph{quite similar} to the physical context of a classical correlation
scenario described by three classical observables $A_{1},$ $A_{2}=B_{1},$ $%
B_{2}$ and discussed in section 3 - with the only substitution of term "%
\emph{classical particles}" by term "\emph{quantum particles}".

Namely, let Alice/Bob joint measurements be performed by identical
measurement devices with settings $a_{1},$ $a_{2}=b_{1},$ $b_{2}$ and upon a
pair of quantum particles, that are indistinguishable, identically prepared
and do not interact with each other during measurements. Then, due to the 
\emph{physical indistinguishability} of sites and measurements of Alice and
Bob specified by the same setting $a_{2}=b_{1},$ these quantum measurements
are described by the same operator averages $\dint \lambda _{1}\mathrm{M}%
_{1}^{(a_{2}=b_{1})}(\mathrm{d}\lambda _{1})=\dint \lambda _{2}\mathrm{M}%
_{2}^{(b_{1})}(\mathrm{d}\lambda _{2})$ and, therefore, in view of relations
(\ref{26}), should be in average modeled by the same quantum observable $%
A_{2}=B_{1}$ at both sites.

\subsection{Examples}

In this section, we present examples of bipartite quantum states, separable
and nonseparable, specified by theorem 3 above. Satisfying inequality (\ref%
{40}) for arbitrary three quantum observables $A_{1},$ $A_{2}=B_{1},$ $B_{2}$%
, with eigenvalues in $[-1,1],$ these bipartite quantum states do not need
to exhibit perfect correlations and may even have a negative correlation
function $\langle \lambda _{1}\lambda _{2}\rangle _{\rho }^{(B_{1},B_{2})}$
whenever the same quantum observable is measured at both sites.

\begin{enumerate}
\item Every state $\rho $ on $\mathcal{H}\otimes \mathcal{H},$ reduced from
a symmetric density operator $\sigma $ on $\mathcal{H}\otimes \mathcal{H}%
\otimes \mathcal{H};$

\item As it is proved by theorem 3 in \cite{7}, every Werner state \cite{14}:%
\begin{equation}
W_{d}(\Phi )=\frac{1+\Phi }{2}\frac{\mathrm{P}_{d}^{(+)}}{r_{d}^{(+)}}+\frac{%
1-\Phi }{2}\frac{\mathrm{P}_{d}^{(-)}}{r_{d}^{(-)}},\text{ \ \ \ }\Phi \in
\lbrack -1,1],  \label{42}
\end{equation}%
on $\mathbb{C}^{d}\otimes \mathbb{C}^{d}$ for $d\geq 3,$ separable ($\Phi
\in \lbrack 0,1]$) or nonseparable ($\Phi \in \lbrack -1,0)$), and every
separable Werner state $W_{2}(\Phi ),$ $\Phi \in \lbrack 0,1]$ on $\mathbb{C}%
^{2}\otimes \mathbb{C}^{2}.$

Here, $\mathrm{P}_{d}^{(\pm )}$ are the orthogonal projections onto the
symmetric (plus sign) and antisymmetric (minus sign) subspaces of $\mathbb{C}%
^{d}\otimes \mathbb{C}^{d}$ with dimensions $r_{d}^{(\pm )}=\mathrm{tr}%
[P_{d}^{(\pm )}]=\frac{d(d\pm 1)}{2},$ respectively.

\item As it is proved by theorem 2 in \cite{8}, each of the noisy states on $%
\mathbb{C}^{d}\otimes \mathbb{C}^{d}:$%
\begin{eqnarray}
\eta _{\psi }(\beta ) &=&\beta \text{ }|\psi \rangle \langle \psi |\text{ }+%
\text{ }(1-\beta )\frac{I_{\mathbb{C}^{d}\otimes \mathbb{C}^{d}}}{d^{2}},%
\text{ \ \ \ }\beta \in \lbrack 0,\frac{1}{2\mathbf{\gamma }_{\psi }^{3}+1}],
\label{43} \\
\mathbf{\gamma }_{\psi } &:&=d\left\Vert \mathrm{tr}_{\mathbb{C}%
^{d}}^{(1)}[|\psi \rangle \langle \psi |]\right\Vert =d\left\Vert \mathrm{tr}%
_{\mathbb{C}^{d}}^{(2)}[|\psi \rangle \langle \psi |]\right\Vert \geq 1, 
\notag
\end{eqnarray}%
corresponding to a pure state $|\psi \rangle \langle \psi |$. Here, $%
\left\Vert \cdot \right\Vert $ means the operator norm and we take into
account that, for any pure state on $\mathbb{C}^{d}\otimes \mathbb{C}^{d},$
the eigenvalues (hence, the operator norms) of the reduced states $\mathrm{tr%
}_{\mathbb{C}^{d}}^{(1)}[|\psi \rangle \langle \psi |]$ and $\mathrm{tr}_{%
\mathbb{C}^{d}}^{(2)}[|\psi \rangle \langle \psi |]$ on $\mathbb{C}^{d}$
coincide.

In particular, the separable noisy singlet\footnote{%
Here, $\{e_{n},$ $n=1,2\}$ is an orthonormal basis in $\mathbb{C}^{2}.$}%
\begin{eqnarray}
\eta _{\psi _{S}}(\beta ) &=&\beta \text{ }|\psi _{S}\rangle \langle \psi
_{S}|\text{ }+\text{ }(1-\beta )\frac{I_{\mathbb{C}^{2}\otimes \mathbb{C}%
^{2}}}{4},\text{ \ \ \ }\beta \in (0,\text{ }\frac{1}{3}\text{ }],
\label{44} \\
\psi _{S} &=&\frac{1}{\sqrt{2}}(e_{1}\otimes e_{2}-e_{2}\otimes e_{1}), 
\notag
\end{eqnarray}%
for which the above parameter $\mathbf{\gamma }_{\psi _{S}}=1.$ This noisy
state constitutes the two-qubit Werner state $W_{2}((1-3\beta )/2).$

Note that whenever the spin observable $\sigma _{n}$ along an arbitrary
direction $n$ in $\mathbb{R}^{3}$ is measured in state $\eta _{\psi
_{s}}(\beta )$ at both sites, then the correlation function is negative: 
\begin{equation}
\langle \lambda _{1}\lambda _{2}\rangle _{\eta _{\psi _{s}}(\beta
)}^{(\sigma _{n},\text{ }\sigma _{n})}=\mathrm{tr}[\eta _{\psi _{s}}(\beta
)\{\sigma _{n}\otimes \sigma _{n})\}]=-\beta <0.  \label{45}
\end{equation}%
This rules out perfect correlations. Furthermore, if Alice's and Bob's spin
measurements on state $\eta _{\psi _{s}}(\beta )$ are projective, then,
given an outcome, say of Alice, the conditional probability that Bob
observes a different outcome is equal to $(1+\beta )/2$ while the
conditional probability that Bob observes the same outcome is $(1-\beta )/2.$
\end{enumerate}

\section{Conclusion}

In the present paper, we introduce a \emph{new condition} sufficient for the
validity of the original Bell inequality in an LHV frame. This LHV condition
is more general than the assumption on perfect correlations or
anticorrelations and incorporates the latter assumption only as a particular
case. For dichotomic bipartite measurements, the new general LHV condition
can be tested experimentally.

Specified for a quantum bipartite case, the new general LHV condition
reduces to the form that does not necessarily imply any correlation between
Alice's and Bob's outcomes if the same quantum observable is measured at
both sites and leads to the existence of the whole class of bipartite
quantum states, separable and nonseparable, that admit an LHV description
under any bipartite quantum correlation scenario with two settings per site
and never violate the "perfect correlation" version of the original Bell
inequality - though do not necessarily exhibit perfect correlations and may
even have a negative correlation function if the same quantum observable is
measured at both sites. Separable and nonseparable bipartite quantum states
specified by us earlier in [5 - 8] are included into the new state class
introduced in the present paper only as a particular subclass.

Our comparative analysis of classical and quantum measurement situations
indicates that an arbitrary classical state $\pi $ satisfies inequality 
\begin{equation}
\left\vert \langle \lambda _{1}\lambda _{2}\rangle _{\pi
}^{(A_{1},B_{1})}-\langle \lambda _{1}\lambda _{2}\rangle _{\pi
}^{(A_{1},B_{2})}\right\vert \leq 1-\langle \lambda _{1}\lambda _{2}\rangle
_{\pi }^{(B_{1},B_{2})}  \label{46}
\end{equation}%
under any Alice's and Bob's measurements, ideal or non-ideal, of arbitrary
classical observables $A_{1},$ $A_{2}=B_{1},$ $B_{2}$ with values in $[-1,1]$
whereas an arbitrary separable\footnote{%
Any separable quantum state admits an LHV description.} quantum state $\rho $
\emph{does not need }to satisfy inequality%
\begin{equation}
\left\vert \langle \lambda _{1}\lambda _{2}\rangle _{\rho
}^{(A_{1},B_{1})}-\langle \lambda _{1}\lambda _{2}\rangle _{\rho
}^{(A_{1},B_{2})}\right\vert \leq 1-\langle \lambda _{1}\lambda _{2}\rangle
_{\rho }^{(B_{1},B_{2})}  \label{47}
\end{equation}%
under any Alice's and Bob's measurements, ideal or non-ideal, of arbitrary
quantum observables $A_{1},$ $A_{2}=B_{1},$ $B_{2}$, with eigenvalues in $%
[-1,1]$.

This, in particular, means\footnote{%
See the discussions at the ends of sections 3, 4.} that, under a $2\times 2$
- setting correlation experiment, performed at both sites by identical
measurement devices with settings $a_{1},$ $a_{2}=b_{1},$ $b_{2}$ and upon a
pair of non-interacting, identically prepared, indistinguishable physical
particles, quantum or classical, the "perfect correlation" version of the
original Bell inequality does not need to hold in an arbitrary separable
quantum case but is always fulfilled in every classical case,\emph{\ }that
is, for any initial state of classical particles and any type of Alice's and
Bob's classical measurements, ideal (necessarily exhibiting perfect
correlations if the same classical observables is measured at both sites) or
non-ideal (not necessarily exhibiting perfect correlations).

Thus, under classical and quantum correlation experiments with the same
physical context, a classical state and a separable quantum state may
exhibit statistically different correlations and \emph{the original Bell
inequality reveals a gap between classicality and quantum separability. }%
This observation agrees with our arguments in \cite{5, 15} that an arbitrary
separable quantum state does not need to satisfy every probabilistic
constraint\ inherent to bipartite measurements on a classical system and
also with the statement of Ollivier and Zurek: "absence of entanglement does
not imply classicality" \cite{16}, built up on the notion of a quantum
discord.

The results of the present paper \emph{disprove} in rigorous mathematical
terms the faulty claims of Simon \cite{3} and Zukowski \cite{4} that, in any
bipartite case, classical or quantum, the perfect correlation version of the
original Bell inequality holds only under the assumption on "perfect
correlations if the same observable is measured at both sites" (\cite{3},
abstract), and that, for the validity of the original Bell inequality in an
LHV frame, the assumption on perfect correlations or anticorrelations is
"minimal" (\cite{4}, page 544 ) and without this assumption "the original
Bell inequality cannot be derived" (\cite{4}, page 544).\medskip

We note that it is specifically due to the latter misconception that the
original Bell inequality has been disregarded in physical applications. Our
results, however, indicate that, for a variety of quantum states, not
necessarily exhibiting any correlation between Alice's and Bob's outcomes
whenever the same quantum observable is measured at both sites, the "perfect
correlation" version of the original Bell inequality holds for any three
bounded quantum observables $A_{1},A_{2}=B_{1},$ $B_{2}$ measured ideally or
non-ideally at sites of Alice and Bob and that, in contrast to the
Clauser-Horne-Shimony-Holt (CHSH) inequality \cite{17} that does not
distinguish between classicality and quantum separability, the original Bell
inequality does distinguish between these two physical concepts.

\section{Appendix}

\begin{proof}[Proof of theorem 2]
Let a state $\rho $ have a source operator $R_{\blacktriangleright }\overset{%
\otimes }{>}0.$ Then, due to Eqs. (\ref{24}), (\ref{29}), the joint
distributions $P_{\rho }^{(a_{i},b_{1})},$ $P_{\rho }^{(a_{i},b_{2})},$ $%
\forall i=1,2,$ admit representations:%
\begin{eqnarray}
P_{\rho }^{(a_{i},b_{1})}(\mathrm{d}\lambda _{1}\times \mathrm{d}\lambda
_{2}) &=&\mathrm{tr}[R_{\blacktriangleright }\{\mathrm{M}_{1}^{(a_{i})}(%
\mathrm{d}\lambda _{1})\otimes \mathrm{M}_{2}^{(b_{1})}(\mathrm{d}\lambda
_{2})\otimes \mathbb{I}_{\mathcal{H}_{2}}\}],  \TCItag{A1} \\
P_{\rho }^{(a_{i},b_{2})}(\mathrm{d}\lambda _{1}\times \mathrm{d}\lambda
_{2}) &=&\mathrm{tr}[R_{\blacktriangleright }\{\mathrm{M}_{1}^{(a_{i})}(%
\mathrm{d}\lambda _{1})\otimes \mathbb{I}_{\mathcal{H}_{2}}\otimes \mathrm{M}%
_{2}^{(b_{2})}(\mathrm{d}\lambda _{2})\}],  \notag
\end{eqnarray}%
where $\mathrm{M}_{1}^{(a_{i})}(\Lambda )=\mathbb{I}_{\mathcal{H}_{1}},$ $%
\mathrm{M}_{2}^{(b_{k})}(\Lambda )=\mathbb{I}_{\mathcal{H}_{2}},$ $i,k=1,2.$
From these representations and property (\ref{30-1}) it follows that, for
any index $i=1,2,$ distributions $P_{\rho }^{(a_{i},b_{1})},$ $P_{\rho
}^{(a_{i},b_{2})}$ are marginals of the normalized joint probability measure 
\begin{equation}
\tau _{R_{\blacktriangleright }}^{(i)}(\mathrm{d}\lambda _{1}\times \mathrm{d%
}\lambda _{2}\times \mathrm{d}\lambda _{2}^{\prime }):=\mathrm{tr}%
[R_{\blacktriangleright }\{\mathrm{M}_{1}^{(a_{i})}(\mathrm{d}\lambda
_{1})\otimes \mathrm{M}_{2}^{(b_{1})}(\mathrm{d}\lambda _{2})\otimes \mathrm{%
M}_{2}^{(b_{2})}(\mathrm{d}\lambda _{2}^{\prime })\}],  \tag{A2}
\end{equation}%
and measures $\tau _{R_{\blacktriangleright }}^{(1)},$ $\tau
_{R_{\blacktriangleright }}^{(2)}$ are compatible in the sense:%
\begin{eqnarray}
\tau _{R_{\blacktriangleright }}^{(1)}(\Lambda \times \mathrm{d}\lambda
_{2}\times \mathrm{d}\lambda _{2}^{\prime }) &=&\mathrm{tr}%
[R_{\blacktriangleright }\{\mathbb{I}_{\mathcal{H}_{1}}\otimes \mathrm{M}%
_{2}^{(b_{1})}(\mathrm{d}\lambda _{2})\otimes \mathrm{M}_{2}^{(b_{2})}(%
\mathrm{d}\lambda _{2}^{\prime })\}]  \TCItag{A3} \\
&=&\tau _{R_{\blacktriangleright }}^{(2)}(\Lambda \times \mathrm{d}\lambda
_{2}\times \mathrm{d}\lambda _{2}^{\prime }).  \notag
\end{eqnarray}%
Due to theorem 2 in [9], the existence of compatible probability measures $%
\tau _{R_{\blacktriangleright }}^{(i)},$ $i=1,2,$ each returning
distributions $P_{\rho }^{(a_{i},b_{1})},$ $P_{\rho }^{(a_{i},b_{2})}$ as
marginals, implies the existence for this correlation experiment of the LHV
model (\ref{18}) where distribution (\ref{22}) takes the form:%
\begin{equation}
\mu (\mathrm{d}\lambda _{1}^{\prime }\times \mathrm{d}\lambda _{2}\times 
\mathrm{d}\lambda _{2}^{\prime })=\tau _{R_{\blacktriangleright }}^{(2)}(%
\mathrm{d}\lambda _{1}^{\prime }\times \mathrm{d}\lambda _{2}\times \mathrm{d%
}\lambda _{2}^{\prime }).  \tag{A4}
\end{equation}%
Substituting this expression into the left-hand side of Eq. (\ref{21}) and
taking into account Eqs. (A2), (\ref{32-1}), (\ref{x}), we derive the
relations:%
\begin{eqnarray}
&&\dint \lambda _{2}^{^{\prime }}(\lambda _{2}\mp \lambda _{1}^{\prime })%
\text{ }\mu (\mathrm{d}\lambda _{1}^{\prime }\times \mathrm{d}\lambda
_{2}\times \mathrm{d}\lambda _{2}^{\prime })  \TCItag{A5} \\
&=&\dint \lambda _{2}^{\prime }(\lambda _{2}\mp \lambda _{1}^{\prime })\text{
}\mathrm{tr}[R_{\blacktriangleright }\{\mathrm{M}_{1}^{(a_{2})}(\mathrm{d}%
\lambda _{1}^{\prime })\otimes \mathrm{M}_{2}^{(b_{1})}(\mathrm{d}\lambda
_{2})\otimes \mathrm{M}_{2}^{(b_{2})}(\mathrm{d}\lambda _{2}^{\prime })\}] 
\notag \\
&=&\mathrm{tr}[\sigma _{R_{\blacktriangleright }}\{B_{1}\otimes B_{2}\}]\mp 
\mathrm{tr}[\rho \{A_{2}\otimes B_{2}\}]\geq 0,  \notag
\end{eqnarray}%
meaning the validity of condition (\ref{21}). This proves the statement.
\end{proof}

\end{document}